\documentclass{article}

\usepackage{times}
\usepackage{amssymb,amsmath,amsthm}
\usepackage{mathrsfs}
\usepackage{epsfig}
\usepackage{color}

\newcommand\ie{{\em i.e.}~}

\newcommand\etc{{\em etc.}}

\newcommand\apriori{{\em a priori}~}

\def\A{\mathcal A}

\def\B{\mathscr B}

\def\C{\mathbb C}

\def\D{\mathscr D}

\def\G{\mathcal G}
\def\H{\mathcal H}

\def\K{\mathscr K}

\def\N{\mathbb N}

\def\R{\mathbb R}
\def\S{\mathscr S}

\def\vg{v_-}
\def\vd{v_+}

\def\dom{\mathcal D}

\def\ltwo{\mathsf{L}^{\:\!\!2}}

\def\e{\mathop{\mathrm{e}}\nolimits}

\def\slim{\mathop{\hbox{\rm s-}\lim}\nolimits}

\def\re{\mathop{\mathsf{Re}}\nolimits}
\def\id{\mathrm{id}_\mathbb{R}}
\def\Ran{\mathop{\mathsf{Ran}}\nolimits}
\def\Ker{\mathop{\mathsf{Ker}}\nolimits}

\newtheorem{Theorem}{Theorem}[section]
\newtheorem{Remark}[Theorem]{Remark}
\newtheorem{Lemma}[Theorem]{Lemma}
\newtheorem{Assumption}[Theorem]{Assumption}
\newtheorem{Corollary}[Theorem]{Corollary}
\newtheorem{Proposition}[Theorem]{Proposition}

\begin{document}


\title{A few results on Mourre theory in a two-Hilbert spaces setting}

\author{S. Richard$^1$\footnote{On leave from Universit\'e de Lyon; Universit\'e
Lyon 1; CNRS, UMR5208, Institut Camille Jordan, 43 blvd du 11 novembre 1918, F-69622
Villeurbanne-Cedex, France.
Supported by the Japan Society for the Promotion of Science (JSPS) and by
``Grants-in-Aid for scientific Research''.}~~and R. Tiedra de
Aldecoa$^2$\footnote{Supported by the Fondecyt Grant 1090008 and by the Iniciativa
Cientifica Milenio ICM P07-027-F ``Mathematical Theory of Quantum and Classical
Magnetic Systems''.}}

\date{\small}
\maketitle \vspace{-1cm}

\begin{quote}
\emph{
\begin{itemize}
\item[$^1$] Graduate School of Pure and Applied Sciences,
University of Tsukuba, \\
1-1-1 Tennodai,
Tsukuba, Ibaraki 305-8571, Japan
\item[$^2$] Facultad de Matem\'aticas, Pontificia Universidad Cat\'olica de Chile,\\
Av. Vicu\~na Mackenna 4860, Santiago, Chile
\item[] \emph{E-mails:} richard@math.univ-lyon1.fr, rtiedra@mat.puc.cl
\end{itemize}
}
\end{quote}


\begin{abstract}
We introduce a natural framework for dealing with Mourre theory in an abstract
two-Hilbert spaces setting. In particular a Mourre estimate for a pair of
self-adjoint operators $(H,A)$ is deduced from a similar estimate for a pair of
self-adjoint operators $(H_0,A_0)$ acting in an auxiliary Hilbert space. A new
criterion for the completeness of the wave operators in a two-Hilbert spaces setting
is also presented.
\end{abstract}


\smallskip



\section{Introduction}
\setcounter{equation}{0}

It is commonly accepted that Mourre theory is a very powerful tool in spectral and
scattering theory for self-adjoint operators. In particular, it naturally leads to
limiting absorption principles which are essential when studying the absolutely
continuous part of self-adjoint operators. Since the pioneering work of E.~Mourre
\cite{M80}, a lot of improvements and extensions have been proposed, and the theory
has led to numerous applications. However, in most of the corresponding works,
Mourre theory is presented in a one-Hilbert space setting and perturbative arguments
are used within this framework. In this paper, we propose to extend the theory to a
two-Hilbert spaces setting and present some results in that direction. In
particular, we show how a Mourre estimate can be deduced for a pair of self-adjoint
operators $(H,A)$ in a Hilbert space $\H$ from a similar estimate for a pair of
self-adjoint operators $(H_0,A_0)$ in a auxiliary Hilbert space $\H_0$.

The main idea of E. Mourre for obtaining results on the spectrum $\sigma(H)$ of a
self-adjoint operator $H$ in a Hilbert space $\H$ is to find an auxiliary
self-adjoint operator $A$ in $\H$ such that the commutator $[iH,A]$ is positive when
localised in the spectrum of $H$. Namely, one looks for a subset
$I\subset\sigma(H)$, a number $a\equiv a(I)>0$ and a compact operator $K\equiv K(I)$
in $\H$ such that
\begin{equation}\label{Mourre}
E^H(I)[iH,A]E^H(I)\ge aE^H(I)+K,
\end{equation}
where $E^H(I)$ is the spectral projection of $H$ on $I$. Such an estimate is commonly
called a Mourre estimate. In general, this positivity condition is obtained via
perturbative technics. Typically, $H$ is a perturbation of a simpler operator $H_0$
in $\H$ for which the commutator $[iH_0,A]$ is easily computable and the positivity
condition easily verifiable. In such a case, the commutator of the formal difference
$H-H_0$ with $A$ can be considered as a small perturbation of $[iH_0,A]$, and one can
still infer the necessary positivity of $[iH,A]$.

In many other situations one faces the problem that $H$ is not the perturbation of
any simpler operator $H_0$ in $\H$. For example, if $H$ is the Laplace-Beltrami
operator on a non-compact manifold, there is no candidate for a simpler operator
$H_0$! Alternatively, for multichannel scattering systems, there might exist more
than one single candidate for $H_0$, and one has to take this multiplicity into
account. In these situations, it is therefore unclear from the very beginning wether
one can find a suitable conjugate operator $A$ for $H$ and how some positivity of
$[iH,A]$ can be deduced from a hypothetic similar condition involving a simpler
operator $H_0$. Of course, these interrogations have found positive answers in
various situations. Nevertheless, it does not seem to the authors that any general
framework has yet been proposed.

The starting point for our investigations is the scattering theory in the two-Hilbert
spaces setting. In this setup, one has a self-adjoint operator $H$ in a Hilbert
space $\H$, and one looks for a simpler self-adjoint operator $H_0$ in an auxiliary
Hilbert space $\H_0$ and a bounded operator $J:\H_0\to\H$ such that the strong limits
$$
\slim_{t\to\pm\infty}\e^{itH}J\e^{-itH_0}\varphi
$$
exist for suitable vectors $\varphi\in\H_0$. If such limits exist for enough
$\varphi\in\H_0$, then some information on the spectral nature of $H$ can be inferred
from similar information on the spectrum of $H_0$. We refer to the books \cite{BW83}
and \cite{Yaf92} for general presentations of scattering theory in the two-Hilbert spaces
setting. Therefore, the following question naturally arises: If $A_0$ is a
conjugate operator for $H_0$ such that \eqref{Mourre} holds with $(H_0,A_0)$ instead
of $(H,A)$, can we define a conjugate operator $A$ for $H$ such that \eqref{Mourre}
holds\;\!? Under suitable conditions, the answer is ``yes'', and its justification is
the content of this paper. In fact, we present a general framework in which a Mourre
estimate for a pair $(H,A)$ can be deduced from a similar Mourre estimate for a pair
$(H_0,A_0)$. In that framework, we suppose the operators $A_0$ and $A$ given
\apriori, and then exhibit sufficient conditions on the formal commutators $[iH,A]$
and $[iH_0,A_0]$ guaranteeing the existence of a Mourre estimate for $(H,A)$ if a
Mourre estimate for $(H_0,A_0)$ is verified (see the assumptions of Theorem
\ref{fonctionrho}). We also show how a conjugate operator $A$ for $H$ can be
constructed from a conjugate operator $A_0$ for $H_0$.

Let us finally sketch the organisation of the paper. In Section \ref{Sec_one}, we
recall a few definitions (borrowed from \cite[Chap.~7]{ABG}) in relation with Mourre
theory in the usual one-Hilbert space setting. In Section \ref{sec_two}, we state our
main result, Theorem \ref{fonctionrho}, on the obtention of a Mourre estimate for
$(H,A)$ from a similar estimate for $(H_0,A_0)$. A complementary result on  higher
order regularity of $H$ with respect to $A$ is also presented. In the second part of
Section \ref{sec_two}, we show how the assumptions of Theorem \ref{fonctionrho} can
be checked for short-range type and long-range type perturbations (note that the
distinction between short-range type and long-range type perturbations is more subtle
here, since $H_0$ and $H$ do not live in the same Hilbert space). We also show how a
natural candidate for $A$ can be constructed from $A_0$. In Section \ref{Example}, we
illustrate our results with the simple example of one-dimensional
Schr\"odinger operator with steplike potential. A more challenging application on
manifolds will be presented in \cite{IRT11} (many other applications such as curved
quantum waveguides, anisotropic Schr\"odinger operators, spin models, \etc~are also
conceivable). Finally, in Section \ref{Sec_Comp} we prove an auxiliary result on the
completeness of the wave operators in the two-Hilbert spaces setting without assuming
that the initial sets of the wave operators are equal to the subspace
$\H_{\rm ac}(H_0)$ of absolute continuity of $H_0$ (in \cite{BW83} and \cite{Yaf92},
only that case is presented and this situation is sometimes too restrictive as will
be shown for example in \cite{IRT11}).

\section{Mourre theory in the one-Hilbert space setting}\label{Sec_one}
\setcounter{equation}{0}

In this section we recall some definitions related to Mourre theory, such as the
regularity condition of $H$ with respect to $A$, providing a precise meaning to the
commutators mentioned in the Introduction. We refer to \cite[Sec.~7.2]{ABG} for more
information and details.

Let us consider a Hilbert space $\H$ with scalar product
$\langle\;\!\cdot\;\!,\;\!\cdot\;\!\rangle_\H$ and norm $\|\;\!\cdot\;\!\|_\H$. Let
also $H$ and $A$ be two self-adjoint operators in $\H$, with domains $\dom(H)$ and
$\dom(A)$. The spectrum of $H$ is denoted by $\sigma(H)$ and its spectral measure by
$E^H(\;\!\cdot\;\!)$. For shortness, we also use the notation
$E^H(\lambda;\varepsilon):=E^H\big((\lambda-\varepsilon,\lambda+\varepsilon)\big)$
for all $\lambda\in\R$ and $\varepsilon>0$.

The operator $H$ is said to be of class $C^1(A)$ if there exists
$z\in\C\setminus\sigma(H)$ such that the map
\begin{equation}\label{C1}
\R\ni t\mapsto\e^{-itA}(H-z)^{-1}\e^{itA}\in\B(\H)
\end{equation}
is strongly of class $C^1$ in $\H$. In such a case, the set $\dom(H)\cap\dom(A)$ is a
core for $H$ and the quadratic form
$
\dom(H)\cap\dom(A)\ni\varphi
\mapsto\langle H\varphi,A\varphi\rangle_\H-\langle A\varphi,H\varphi\rangle_\H
$
is continuous in the topology of $\dom(H)$. This form extends then uniquely to a
continuous quadratic form $[H,A]$ on $\dom(H)$, which can be identified with a
continuous operator from $\dom(H)$ to the adjoint space $\dom(H)^*$. Furthermore, the
following equality holds:
$$
\big[A,(H-z)^{-1}\big]=(H-z)^{-1}[H,A](H-z)^{-1}.
$$
This $C^1(A)$-regularity of $H$ with respect to $A$ is the basic ingredient for any
investigation in Mourre theory. It is also at the root of the proof of the Virial
Theorem (see for example \cite[Prop.~7.2.10]{ABG} or \cite{GG99}).

Note that if $H$ is of class $C^1(A)$ and if $\eta \in C^\infty_{\rm c}(\R)$ (the set
of smooth functions on $\R$ with compact support), then the quadratic form
$
\dom(A)\ni\varphi\mapsto\langle\bar\eta(H)\varphi,A\varphi\rangle_\H
-\langle A\varphi,\eta(H)\varphi\rangle_\H
$
also extends uniquely to a continuous quadratic form $[\eta(H)A,]$ on $\H$,
identified with a bounded operator on $\H$.

We now recall the definition of two very useful functions in Mourre theory described
in \cite[Sec.~7.2]{ABG}. For that purpose, we use the following notations: for two
bounded operators $S$ and $T$ in a common Hilbert space we write $S\approx T$ if
$S-T$ is compact, and we write $S\lesssim T$ if there exists a compact operator $K$
such that $S\le T+K$. If $H$ is of class $C^1(A)$ and $\lambda\in\R$ we set
$$
\varrho^A_H(\lambda)
:=\sup\big\{a\in\R \mid\exists\;\!\varepsilon>0~\,\hbox{s.t.}~\,a\;\!
E^H(\lambda;\varepsilon)\le E^H(\lambda;\varepsilon)[iH,A]
E^H(\lambda;\varepsilon)\big\}.
$$
A second function, more convenient in applications, is
$$
\widetilde\varrho^A_H(\lambda)
:=\sup\big\{a\in\R \mid\exists\;\!\varepsilon>0~\,\hbox{s.t.}~\,a\;\!
E^H(\lambda;\varepsilon)\lesssim E^H(\lambda;\varepsilon)[iH,A]
E^H(\lambda;\varepsilon)\big\}.
$$
Note that the following equivalent definition is often useful:
\begin{equation}\label{autre}
\widetilde\varrho^A_H(\lambda)
=\sup\big\{a\in\R\mid\exists\;\!\eta\in C^\infty_{\rm c}(\R)
\hbox{ real}~\,\hbox{s.t.}~\,\eta(\lambda)\neq0,~a\;\!\eta(H)^2\lesssim
\eta(H)[iH,A]\eta(H)\big\}.
\end{equation}
It is commonly said that $A$ is conjugate to $H$ at the point $\lambda\in\R$ if
$\widetilde\varrho^A_H(\lambda)>0$, and that $A$ is strictly conjugate to $H$ at
$\lambda$ if $\varrho^A_H(\lambda)>0$. Furthermore, the function
$\widetilde\varrho^A_H:\R\to(-\infty,\infty]$ is lower semicontinuous and satisfies
$\widetilde\varrho^A_H(\lambda)<\infty$ if and only if $\lambda$ belongs to the
essential spectrum $\sigma_{\rm ess}(H)$ of $H$. One also has
$\widetilde\varrho^A_H(\lambda)\ge\varrho^A_H(\lambda)$ for all $\lambda\in\R$.

Another property of the function $\widetilde\varrho$, often used in the one-Hilbert space
setting, is its stability under a large class of perturbations: Suppose that $H$ and
$H'$ are self-adjoint operators in $\H$ and that both operators $H$ and $H'$ are of
class $C^1_{\rm u}(A)$, \ie such that the map \eqref{C1} is $C^1$ in norm. Assume
furthermore that the difference $(H-i)^{-1}-(H'-i)^{-1}$ belongs to $\K(\H)$, the
algebra of compact operators on $\H$. Then, it is proved in \cite[Thm.~7.2.9]{ABG}
that $\widetilde\varrho^A_{H'}=\widetilde\varrho^A_H$, or in other words that $A$ is
conjugate to $H'$ at a point $\lambda\in\R$ if and only if $A$ is conjugate to $H$ at
$\lambda$.

Our first contribution in this paper is to extend such a result to the two-Hilbert spaces
setting. But before this, let us recall the importance of the set
$\widetilde\mu^A(H)\subset\R$ on which $\widetilde\varrho^A_H(\;\!\cdot\;\!)>0$: if
$H$ is slightly more regular than $C^1(A)$, then $H$ has locally at most a finite
number of eigenvalues on $\widetilde\mu^A(H)$ (multiplicities counted), and $H$ has
no singularly continuous spectrum on $\widetilde\mu^A(H)$ (see \cite[Thm.~7.4.2]{ABG}
for details).

\section{Mourre theory in the two-Hilbert spaces setting}\label{sec_two}
\setcounter{equation}{0}

From now on, apart from the triple $(\H,H,A)$ of Section \ref{Sec_one}, we consider
a second triple $(\H_0,H_0,A_0)$ and an identification operator $J:\H_0\to\H$. The
existence of two such triples is quite standard in scattering theory, at least for
the pairs $(\H,H)$ and $(\H_0,H_0)$ (see for instance the books \cite{BW83,Yaf92}).
Part of our goal in what follows is to show that the existence of the conjugate
operators $A$ and $A_0$ is also natural, as was realised in the context of scattering
on manifolds \cite{IRT11}.

So, let us consider a second Hilbert space $\H_0$ with scalar product
$\langle\;\!\cdot\;\!,\;\!\cdot\;\!\rangle_{\H_0}$ and norm
$\|\;\!\cdot\;\!\|_{\H_0}$. Let also $H_0$ and $A_0$ be two self-adjoint operators in
$\H_0$, with domains $\dom(H_0)$ and $\dom(A_0)$. Clearly, the $C^1(A_0)$-regularity
of $H_0$ with respect to $A_0$ can be defined as before, and if $H_0$ is of class
$C^1(A_0)$ then the definitions of the two functions $\varrho^{A_0}_{H_0}$ and
$\widetilde \varrho^{A_0}_{H_0}$ hold as well.

In order to compare the two triples, it is natural to require the existence of a map
$J\in\B(\H_0,\H)$ having some special properties (for example, the ones needed for
the completeness of the wave operators, see Section \ref{Sec_Comp}). But for the time
being, no additional information on $J$ is necessary. In the one-Hilbert space setting,
the operator $H$ is typically a perturbation of the simpler operator $H_0$. And as
mentioned above, the stability of the function $\widetilde\varrho_{H_0}^{A_0}$ is an
efficient tool to infer information on $H$ from similar information on $H_0$. In the
two-Hilbert spaces setting, we are not aware of any general result allowing the
computation of the function $\widetilde\varrho_H^A$ in terms of the function
$\widetilde\varrho_{H_0}^{A_0}$. The obvious reason for this being the impossibility
to consider $H$ as a direct perturbation of $H_0$ since these operators do not live
in the same Hilbert space. Nonetheless, the next theorem gives a result in that
direction:

\begin{Theorem}\label{fonctionrho}
Let $(\H,H,A)$ and $(\H_0,H_0,A_0)$ be as above, and assume that
\begin{enumerate}
\item[(i)] the operators $H_0$ and $H$ are of class $C^1(A_0)$ and $C^1(A)$,
respectively,
\item[(ii)] for any $\,\eta\in C^\infty_{\rm c}(\R)$ the difference of bounded
operators $J[iA_0,\eta(H_0)]J^*-[iA,\eta(H)]$ belongs to $\K(\H)$,
\item[(iii)] for any $\,\eta\in C^\infty_{\rm c}(\R)$ the difference
$J\eta(H_0)-\eta(H)J$ belongs to $\K(\H_0,\H)$,
\item[(iv)] for any $\,\eta\in C^\infty_{\rm c}(\R)$ the operator
$\eta(H)(JJ^*-1)\eta(H)$ belongs to $\K(\H)$.
\end{enumerate}
Then, one has $\widetilde\varrho_H^A\ge\widetilde\varrho_{H_0}^{A_0}$. In particular,
if $A_0$ is conjugate to $H_0$ at $\lambda\in\R$, then $A$ is conjugate to $H$ at
$\lambda$.
\end{Theorem}

Note that with the notations introduced in the previous section, Assumption (ii)
reads $J[iA_0,\eta(H_0)]J^*\approx[iA,\eta(H)]$. Furthermore, since the vector space
generated by the family of functions $\{(\;\cdot\;-z)^{-1}\}_{z\in\C\setminus\R}$ is
dense in $C_0(\R)$ and the set $\K(\H_0,\H)$ is closed in $\B(\H_0,\H)$, the
condition $J(H_0-z)^{-1}-(H-z)^{-1}J \in\K(\H_0,\H)$ for all $z\in\C\setminus\R$
implies Assumption (iii) (here, $C_0(\R)$ denotes the set of continuous functions on
$\R$ vanishing at $\pm\infty$).

\begin{proof}
Let $\eta\in C^\infty_{\rm c}(\R;\R)$, and define
$\eta_1,\eta_2\in C^\infty_{\rm c}(\R;\R)$ by $\eta_1(x):=x\;\!\eta(x)$ and
$\eta_2(x):=x\;\!\eta(x)^2$. Under Assumption (i), it is shown in
\cite[Eq.~7.2.18]{ABG}
that
$$
\eta(H)[iA,H]\eta(H)=[iA,\eta_2(H)]-2\re\big\{[iA,\eta(H)]\eta_1(H)\big\}.
$$
Therefore, one infers from Assumptions (ii) and (iii) that
\begin{align*}
&\eta(H)[iA,H]\eta(H)\\
&\approx J[iA_0,\eta_2(H_0)]J^*-2\re\big\{J[iA_0,\eta(H_0)]J^*\eta_1(H)\big\}\\
&=J[iA_0,\eta_2(H_0)]J^*-2\re\big\{J[iA_0,\eta(H_0)]\eta_1(H_0)J^*\big\}
-2\re\big\{J[iA_0,\eta(H_0)]\big(J^*\eta_1(H)-\eta_1(H_0)J^*\big)\big\}\\
&\approx J[iA_0,\eta_2(H_0)]J^*-2J\re\big\{[iA_0,\eta(H_0)]\eta_1(H_0)\big\}J^*\\
&=J\eta(H_0)[iA_0,H_0]\eta(H_0)J^*,
\end{align*}
which means that
\begin{equation}\label{Eone}
\eta(H)[iA,H]\eta(H)\approx J\eta(H_0)[iA_0,H_0]\eta(H_0)J^*.
\end{equation}
Furthermore, if $a\in\R$ is such that
$\eta(H_0)[iA_0,H_0]\eta(H_0)\gtrsim a \eta(H_0)^2$, then Assumptions (iii) and (iv)
imply that
\begin{equation}\label{Etwo}
J\eta(H_0)[iA_0,H_0]\eta(H_0)J^*
\gtrsim aJ\eta(H_0)^2J^*
\approx a\eta(H)JJ^*\eta(H)
\approx a\eta(H)^2.
\end{equation}
Thus, one obtains $\eta(H)[iA,H]\eta(H)\gtrsim a\eta(H)^2$ by combining \eqref{Eone}
and \eqref{Etwo}. This last estimate, together with the definition \eqref{autre} of
the functions $\widetilde\varrho_{H_0}^{A_0}$ and $\widetilde\varrho_H^A$ in terms of
the localisation function $\eta$, implies the claim.
\end{proof}

As mentioned in the previous sections, the $C^1(A)$-regularity of $H$ and the Mourre
estimate are crucial ingredients for the analysis of the operator $H$, but they are
in general not sufficient. For instance, the nature of the spectrum of $H$ or the
existence and the completeness of the wave operators is usually proved under a
slightly stronger $C^{1,1}(A)$-regularity condition of $H$. It would certainly be
valuable if this regularity condition could be deduced from a similar information on
$H_0$. Since we have not been able to obtain such a result, we simply refer to
\cite{ABG} for the definition of this class of regularity and present below a coarser
result. Namely, we show that the regularity condition ``$H$ is of class $C^n(A)$''
can be checked by means of explicit computations involving only $H$ and not its
resolvent. For simplicity, we present the simplest, non-perturbative version of the
result; more refined statements involving perturbations as in Sections \ref{S1} and
\ref{S2} could also be proved.

For that purpose, we first recall that $H$ is of class $C^n(A)$ if the map \eqref{C1}
is strongly of class $C^n$. We also introduce the following slightly more general
regularity class: Assume that $(\G,\H)$ is a Friedrichs couple, \ie a pair $(\G,\H)$
with $\G$ a Hilbert space densely and continuously embedded in $\H$. Assume
furthermore that the unitary group $\{\e^{itA}\}_{t\in\R}$ leaves $\G$ invariant, so
that the restriction of this group to $\G$ generates a $C_0$-group, with generator
also denoted by $A$. In such a situation, an operator $T\in\B(\G,\H)$ is said to
belong to $C^n(\A;\G,\H)$ if the map
$$
\R\ni t\mapsto\e^{-itA}T\e^{itA}\in\B(\G,\H)
$$
is strongly of class $C^n$. Similar definitions hold with $T$ in $\B(\H,\G)$, in
$\B(\G,\G)$ or in $\B(\H,\H)$ (in the latter case, one simply writes $T\in C^n(A)$
instead of $T \in C^n(A;\H,\H)$).

The next proposition (which improves slightly the result of \cite[Lemma~1.2]{MW11})
is an extension of \cite[Thm.~6.3.4.(c)]{ABG} to higher orders of regularity of $H$
with respect to $A$. We use for it the notation $\G$ for the domain $\dom(H)$ of $H$
endowed with its natural Hilbert space structure. We also recall that if $H$ is of
class $C^1(A)$, then $[iH,A]$ can be identified with a bounded operator from $\G$ to
$\G^*$. It has been proved in \cite[Lemma~2]{GG99} that if this operator maps $\G$
into $\H$, then $\{\e^{itA}\}_{t\in\R}$ leaves $\G$ invariant, and thus one has a
$C_0$-group in $\G$.

\begin{Proposition}\label{reguln}
Let $H$ be of class $C^1(A)$, assume that $[iH,A]\in\B(\G,\H)$ and suppose that
$[iH,A]\in C^n(A;\G,\H)$ for some integer $n\ge0$. Then
$(H-z)^{-1}\in C^{n+1}(A;\H,\G)\subset C^{n+1}(A)$ for any $z\in\C\setminus\R$.
\end{Proposition}

\begin{proof}
We prove the claim by induction on $n$. For $n=0$, one has
$[iH,A]\in\B(\G,\H)\equiv C^0(A;\G,\H)$. It follows from the equality
\begin{equation}\label{step1}
\big[i(H-z)^{-1},A\big]=-(H-z)^{-1}[iH,A](H-z)^{-1}
\end{equation}
and from the inclusion $(H-z)^{-1}\in\B(\H,\G)$ that
$\big[i(H-z)^{-1},A\big]\in\B(\H,\G)$. Then, one infers that
$(H-z)^{-1}\in C^1(A;\H,\G)$ by using \cite[Prop.~5.1.2.(b)]{ABG}.

Now, assume that the statement is true for $n-1\geq 0$, namely,
$[iH,A]\in C^n(A;\G,\H)$ and $(H-z)^{-1}\in C^n(A;\H,\G)$. Then, by taking into
account account \eqref{step1} and the property of regularity for product of operators
stated in \cite[Prop.~5.1.5]{ABG}, one obtains that
$\big[i(H-z)^{-1},A\big]\in C^n(A;\H,\G)$. This is equivalent to the inclusion
$(H-z)^{-1}\in C^{n+1}(A;\H,\G)$, which proves the statement for $n$.
\end{proof}

Usually, the regularity of $H_0$ with respect to $A_0$ is easy to check. On the other
hand, the regularity of $H$ with respect to $A$ is in general rather difficult to
establish, and various perturbative criteria have been developed for that purpose in
the one-Hilbert space setting. Often, a distinction is made between so-called
short-range and long-range perturbations. Roughly speaking, the difference between
these types perturbations is that the two terms of the formal commutator
$[A,H-H_0]=A(H-H_0)-(H-H_0)A$ are treated separately in the former situation while
the commutator $[A,H-H_0]$ is really computed in the latter situation. In the first
case, one usually requires more decay and less regularity, while in the second case
more regularity but less decay are imposed. Obviously, this distinction cannot be as
transparent in the general two-Hilbert spaces setting presented here. Still, a
certain distinction remains, and thus we dedicate to it the following two
complementary sections.

\subsection{Short-range type perturbations}\label{S1}

We show below how the condition ``$H$ is of class $C^1(A)$'' and the assumptions (ii)
and (iii) of Theorem \ref{fonctionrho} can be verified for a class of short-range
type perturbations. Our approach is to derive information on $H$ from some equivalent
information on $H_0$, which is usually easier to obtain. Accordingly, our results
exhibit some perturbative flavor. The price one has to pay is that a compatibility
condition between $A_0$ and $A$ is necessary. For $z\in\C\setminus\R$, we use the
shorter notations $R_0(z):=(H_0-z)^{-1}$, $R(z):=(H-z)^{-1}$ and
\begin{equation}\label{limpide}
B(z):=JR_0(z)-R(z)J\in\B(\H_0,\H).
\end{equation}

\begin{Proposition}\label{C1short}
Let $H_0$ be of class $C^1(A_0)$ and assume that $\D\subset\H$ is a core for $A$
such that $J^*\D\subset\dom(A_0)$. Suppose furthermore that for any
$z\in\C\setminus\R$
\begin{equation}\label{hyp1}
\overline{B(z)A_0\upharpoonright\dom(A_0)}\in\B(\H_0,\H)
\qquad\hbox{and}\qquad\overline{R(z)(JA_0J^*-A)\upharpoonright\D}\in\B(\H).
\end{equation}
Then, $H$ is of class $C^1(A)$.
\end{Proposition}

\begin{proof}
Take $\psi\in\D$ and $z\in\C\setminus\R$. Then, one gets
\begin{align*}
&\big\langle R(\bar z)\psi,A\psi\big\rangle_\H
-\big\langle A\psi,R(z)\psi \big\rangle_\H\\
&=\big\langle R(\bar z)\psi,A\psi\big\rangle_\H
-\big\langle A\psi,R(z)\psi\big\rangle_\H
-\big\langle\psi,J[R_0(z),A_0]J^*\psi\big\rangle_\H
+\big\langle\psi,J[R_0(z),A_0]J^*\psi\big\rangle_\H\\
&=\big\langle B(\bar z)A_0J^*\psi,\psi,\big\rangle_\H
-\big\langle\psi,B(z)A_0J^*\psi\big\rangle_\H
+\big\langle\psi,J[R_0(z),A_0]J^*\psi\big\rangle_\H\\
&\qquad+\big\langle R(\bar z)(JA_0J^*-A)\psi,\psi\big\rangle_\H
-\big\langle\psi,R(z)(JA_0J^*-A)\psi\big\rangle_\H.
\end{align*}
Now, one has
$$
\big|\big\langle B(\bar z)A_0J^*\psi,\psi,\big\rangle_\H
-\big\langle\psi,B(z)A_0J^*\psi\big\rangle_\H\big|
\le{\rm Const.}\;\!\|\psi\|_\H^2
$$
due to the first condition in \eqref{hyp1}, and one has
$$
\big|\big\langle R(\bar z)(JA_0J^*-A)\psi,\psi\big\rangle_\H
-\big\langle\psi,R(z)(JA_0J^*-A)\psi\big\rangle_\H\big|
\le{\rm Const.}\;\!\|\psi\|^2_\H
$$
due to the second condition in \eqref{hyp1}. Furthermore, since $H_0$ is of class
$C^1(A_0)$ one also has
$$
\big|\big\langle\psi,J[R_0(z),A_0]J^*\psi\big\rangle_\H\big|
\le{\rm Const.}\;\!\|\psi\|^2_\H.
$$
Since $\D$ is a core for $A$, the conclusion then follows from
\cite[Lemma~6.2.9]{ABG}.
\end{proof}

We now show how the assumption (ii) of Theorem \ref{fonctionrho} is verified for a
short-range type perturbation. Note that the hypotheses of the following proposition
are slightly stronger than the ones of Proposition \ref{C1short}, and thus $H$ is
automatically of class $C^1(A)$.

\begin{Proposition}\label{ass2_short}
Let $H_0$ be of class $C^1(A_0)$ and assume that $\D\subset\H$ is a core for $A$ such
that $J^*\D\subset\dom(A_0)$. Suppose furthermore that for any $z\in\C\setminus\R$
\begin{equation}\label{c123}
\overline{B(z)A_0\upharpoonright\dom(A_0)}\in\K(\H_0,\H)
\qquad\hbox{and}\qquad\overline{R(z)(JA_0J^*-A)\upharpoonright\D}\in\K(\H).
\end{equation}
Then, for each $\eta\in C^\infty_{\rm c}(\R)$ the difference of bounded operators
$J[A_0,\eta(H_0)]J^*-[A,\eta(H)]$ belongs to $\K(\H)$.
\end{Proposition}

\begin{proof}
Take $\psi,\psi'\in\D$ and $z\in\C\setminus\R$. Then, one gets from the proof of
Proposition \ref{C1short} that
\begin{align*}
&\big\langle\psi',J[A_0,R_0(z)]J^*\psi\big\rangle_\H
-\big\langle\psi',[A,R(z)]\psi\rangle_\H\\
&=\big\langle B(\bar z) A_0J^*\psi',\psi,\big\rangle_\H
-\big\langle\psi',B(z)A_0J^*\psi\big\rangle_\H\\
&\qquad+\big\langle R(\bar z)(JA_0J^*-A)\psi',\psi\big\rangle_\H
-\big\langle\psi',R(z)(JA_0J^*-A)\psi\big\rangle_\H.
\end{align*}
By the density of $\D$ in $\H$, one then infers from the hypotheses that
$J[A_0,R_0(z)]J^*-[A,R(z)]$ belongs to $\K(\H)$.

To show the same result for functions $\eta\in C^\infty_{\rm c}(\R)$ instead of
$(\;\!\cdot\;\!-z)^{-1}$, one needs more refined estimates. Taking the first
resolvent identity into account one obtains
$$
B(z)=\big\{1+(z-i)R(z)\big\}B(i)\big\{1+(z-i)R_0(z)\big\}.
$$
Thus, one gets on $\D$ the equalities
\begin{equation}\label{llabol}
B(z)A_0 J^*=\big\{1+(z-i)R(z)\big\}B(i)A_0\big\{1+(z-i)R_0(z)\big\}J^*
+\big\{1+(z-i)R(z)\big\}B(i)(z-i)[R_0(z),A_0]J^*,
\end{equation}
where
$$
[R_0(z),A_0]=\big\{1+(z-i)R_0(z)\big\}R_0(i)[A_0,H_0]R_0(i)\big\{1+(z-i)R_0(z)\big\}.
$$
Obviously, these equalities extend to all of $\H$ since they involve only bounded
operators. Letting $z=\lambda+i\mu$ with $|\mu|\le 1$, one even gets the bound
$$
\big\|B(z)A_0J^*\big\|_{\B(\H)}
\le{\rm Const.}\;\!\bigg(1+\frac{|\lambda+i(\mu-1)|}{|\mu|}\bigg)^4.
$$
Furthermore, since the first and second terms of \eqref{llabol} extend to elements of
$\K(\H)$, the third term of \eqref{llabol} also extends to an element of $\K(\H)$.
Similarly, the operator on $\D$
$$
R(z)(JA_0J^*-A)\equiv\big\{1+(z-i)R(z)\big\}R(i)(JA_0J^*-A)
$$
extends to a compact operator in $\H$, and one has the bound
$$
\big\|R(z)(JA_0J^*-A) \big\|_{\B(\H)}
\le{\rm Const.}\;\!\bigg(1+\frac{|\lambda+i(\mu-1)|}{|\mu|}\bigg).
$$

Now, observe that for any $\eta\in C^\infty_{\rm c}(\R)$ and any $\psi,\psi'\in\D$
one has
\begin{align}
&\big\langle\psi',J[A_0,\eta(H_0)]J^*\psi\big\rangle_\H
-\big\langle\psi',[A,\eta(H)]\psi\big\rangle_\H\nonumber\\
&=\big\langle\big\{J\overline\eta(H_0)
-\overline\eta(H)J\big\}A_0 J^*\psi',\psi\big\rangle_\H
-\big\langle\psi',\big\{J\eta(H_0)-\eta(H)J\big\}
A_0J^*\psi\big\rangle_\H.\nonumber\\
&\qquad+\big\langle\overline\eta(H)(JA_0J^*-A)\psi',\psi\big\rangle_\H
-\big\langle\psi',\eta(H)(JA_0J^*-A)\psi\big\rangle_\H.\label{lopp}
\end{align}
Then, by expressing the operators $\eta(H_0)$ and $\eta(H)$ in terms of their
respective resolvents (using for example \cite[Eq.~6.1.18]{ABG}) and by taking the
above estimates into account, one obtains that
$\big\{J\eta(H_0)-\eta(H)J\big\}A_0J^*$ and $\eta(H)(JA_0J^*-A)$ are equal on $\D$ to
a finite sum of norm convergent integrals of compact operators. Since $\D$ is dense
in $\H$, these equalities between bounded operators extend continuously to equalities
in $\B(\H)$, and thus the statement follows by using \eqref{lopp}.
\end{proof}

\begin{Remark}
As mentioned just after Theorem \ref{fonctionrho}, the requirement
$B(z)\in\K(\H_0,\H)$ for all $z\in\C\setminus\R$ implies the assumption (iii) of
Theorem \ref{fonctionrho}. Since an a priori stronger requirement is imposed in the
first condition of \eqref{c123}, it is likely that in applications the compactness
assumption (iii) will follow from the necessary conditions ensuring the first
condition in \eqref{c123}.
\end{Remark}

Before turning to the long-range case, let us reconsider the above statements in the
special situation where $A= JA_0J^*$. This case deserves a particular attention since
it represents the most natural choice of conjugate operator for $H$ when $A_0$ is a
conjugate operator for $H_0$. However, in order to deal with a well-defined
self-adjoint operator $A$, one needs the following assumption:

\begin{Assumption}\label{bobo}
There exists a set $\D\subset\dom(A_0J^*)\subset \H$ such that $JA_0J^*$ is
essentially self-adjoint on $\D$, with corresponding self-adjoint extension denoted
by $A$.
\end{Assumption}

Assumption \ref{bobo} might be difficult to check in general, but in concrete
situations the choice of the set $\D$ can be quite natural. We now show how the
assumptions of the above propositions can easily be checked under Assumption
\ref{bobo}. Recall that the operator $B(z)$ was defined in \eqref{limpide}.

\begin{Corollary}\label{bobone}
Let $H_0$ be of class $C^1(A_0)$, suppose that Assumption \ref{bobo} holds for some
set $\D\subset\H$, and for any $z\in\C\setminus\R$ assume that
$$
\overline{B(z)A_0\upharpoonright\dom(A_0)}\in\B(\H_0,\H).
$$
Then, $H$ is of class $C^1(A)$.
\end{Corollary}

\begin{proof}
All the assumptions of Proposition \ref{C1short} are verified.
\end{proof}

\begin{Corollary}\label{compact}
Let $H_0$ be of class $C^1(A_0)$, suppose that Assumption \ref{bobo} holds for some
set $\D\subset\H$, and for any $z\in\C\setminus\R$ assume that
\begin{equation}\label{cond_cor_comp}
\overline{B(z)A_0\upharpoonright\dom(A_0)}\in\K(\H_0,\H).
\end{equation}
Then, for each $\eta\in C^\infty_{\rm c}(\R)$ the difference of bounded operators
$J[A_0,\eta(H_0)]J^*-[A,\eta(H)]$ belongs to $\K(\H)$.
\end{Corollary}

\begin{proof}
All the assumptions of Proposition \ref{ass2_short} are verified.
\end{proof}

\begin{Remark}
As mentioned above the choice $A = JA_0J^*$ is natural when $A_0$ is a conjugate
operator for $H_0$. With that respect the second conditions in \eqref{hyp1} and
\eqref{c123} quantify how much one can deviate from this natural choice.
\end{Remark}

The most important consequence of Mourre theory is the obtention of a limiting
absorption principle for $H_0$ and $H$. Rather often, the space defined in terms of
$A_0$ (resp. $A$) in which holds the limiting absorption principle for $H_0$ (resp.
$H$) is not adequate for applications. In \cite[Prop.~7.4.4]{ABG} a method is given
for expressing the limiting absorption principle for $H_0$ in terms of an auxiliary
operator $\Phi_0$ in $\H_0$ more suitable than $A_0$. Obviously, this abstract result
also applies for three operators $H$, $A$ and $\Phi$ in $\H$, but one crucial
condition is that $(H-z)^{-1}\dom(\Phi)\subset\dom(A)$ for suitable $z\in\C$. In the
next lemma, we provide a sufficient condition allowing to infer this information from
similar information on the operators  $H_0$, $A_0$ and $\Phi_0$ in $\H_0$. Note that
$\Phi$ does not need to be of the form $J\Phi_0J^*$ but that such a situation often
appears in applications.

\begin{Lemma}\label{passage}
Let $z\in\C\setminus\{\sigma(H_0)\cup\sigma(H)\}$. Suppose that Assumption \ref{bobo}
holds for some set $\D\subset\H$. Assume that
$$
\overline{B(\bar z) A_0\upharpoonright\dom(A_0)}\in\B(\H_0,\H).
$$
Furthermore, let $\Phi_0$ and $\Phi$ be self-adjoint operators in $\H_0$ and $\H$
satisfying $(H_0-z)^{-1}\dom(\Phi_0)\subset\dom(A_0)$ and
$J^*(\Phi-i)^{-1}-(\Phi_0-i)^{-1}J^*=(\Phi_0-i)^{-1}B$ for some $B\in\B(\H,\H_0)$.
Then, one has the inclusion $(H-z)^{-1}\dom(\Phi)\subset\dom(A)$.
\end{Lemma}

\begin{proof}
Let $\psi\in\D$ and $\psi'\in\H$. Then, one has
\begin{align*}
&\big\langle A\psi,(H-z)^{-1}(\Phi-i)^{-1}\psi'\big\rangle_\H\\
&=\big\langle\big\{(H-\bar z)^{-1}J-J(H_0-\bar z)^{-1}\big\}A_0J^*\psi,
(\Phi-i)^{-1}\psi'\big\rangle_\H
+\big\langle J(H_0-\bar z)^{-1}A_0J^*\psi,(\Phi-i)^{-1}\psi'\big\rangle_\H\\
&=-\big\langle B(\bar z)A_0J^*\psi,
(\Phi-i)^{-1}\psi'\big\rangle_\H
+\big\langle(H_0-\bar z)^{-1}A_0J^*\psi,(\Phi_0-i)^{-1}J^*\psi'\big\rangle_{\H_0}\\
&\qquad+\big\langle(H_0-\bar z)^{-1}A_0J^*\psi,
(\Phi_0-i)^{-1}B\psi'\big\rangle_{\H_0}.
\end{align*}
So,
$
\big|\big\langle A\psi,(H-z)^{-1}(\Phi-i)^{-1}\psi'\big\rangle_\H\big|
\le{\rm Const.}\;\!\|\psi\|_\H
$,
and thus $(H-z)^{-1}(\Phi-i)^{-1}\psi'\in\dom(A)$, since $A$ is essentially
self-adjoint on $\D$.
\end{proof}

\subsection{Long-range type perturbations}\label{S2}

In the case of a long-range type perturbation, the situation is slightly less
satisfactory than in the short-range case. One reason comes from the fact that one
really has to compute the commutator $[A,H-H_0]$ instead of treating the terms
$A(H-H_0)$ and $(H-H_0)A$ separately. However, a rather efficient method for checking
that ``$H$ is of class $C^1(A)$'' has been put into evidence in
\cite[Lemma.~A.2]{GM08}. We start by recalling this result and then we propose a
perturbative type argument for checking the assumption (ii) of Theorem
\ref{fonctionrho}. Note that there is a missprint in the hypothesis $1$ of
\cite[Lemma~A.2]{GM08}; the meaningless condition $\sup_n\|\chi_n\|_{\dom(H)}<\infty$
has to be replaced by $\sup_n\|\chi_n\|_{\B(\dom(H))}<\infty$.

\begin{Lemma}[Lemma A.2 of \cite{GM08}]
Let $\D\subset\H$ be a core for $A$ such that  $\D\subset\dom(H)$ and $H\D\subset\D$.
Let $\{\chi_n\}_{n\in\N}$ be a family of bounded operators on $\H$ such that
\begin{enumerate}
\item[(i)] $\chi_n\D\subset\D$ for each $n\in\N$, $\slim_{n\to\infty}\chi_n=1$ and
$\,\sup_n\|\chi_n\|_{\B(\dom(H))}<\infty$,
\item[(ii)] for all $\psi\in\D$, one has $\slim_{n\to\infty}A\chi_n\psi=A\psi$,
\item[(iii)] there exists $z\in\C\setminus\sigma(H)$ such that
$\chi_n R(z)\D\subset\D$ and $\chi_nR(\bar z)\D\subset\D$ for each $n\in\N$,
\item[(iv)] for all $\psi\in\D$, one has
$\slim_{n\to\infty}A[H,\chi_n]R(z)\psi=0$ and
$\slim_{n\to\infty}A[H,\chi_n]R(\bar z)\psi=0$.
\end{enumerate}
Finally, assume that for all $\psi\in\D$
$$
\big|\langle A\psi,H\psi\rangle_\H-\langle H\psi,A\psi\rangle_\H\big|
\le{\rm Const.}\;\!\big(\|H\psi\|^2+\|\psi\|^2\big).
$$
Then, $H$ is of class $C^1(A)$.
\end{Lemma}

In the next statement we provide conditions under which the assumption (ii) of
Theorem \ref{fonctionrho} is verified for a long-range type perturbation. One
condition is that for each $z\in\C\setminus\R$ the operator $B(z)$ belongs to
$\K(\H_0,\H)$, which means that the hypothesis (iii) of Theorem \ref{fonctionrho}
is also automatically satisfied. We stress that no direct relation between $A_0$ and
$A$ is imposed; the single relation linking $A_0$ and $A$ only involves the
commutators $[H_0,A_0]$ and $[H,A]$. On the other hand, the condition on $H_0$ is
slightly stronger than just the $C^1(A_0)$-regularity.

\begin{Proposition}\label{ass2_long}
Let $H_0$ be of class $C^1(A_0)$ with $[H_0,A_0]\in\B\big(\dom(H_0),\H_0\big)$ and
let $H$ be of class $C^1(A)$. Assume that the operator $J\in\B(\H_0,\H)$ extends to
an element of $\B\big(\dom(H_0)^*,\dom(H)^*\big)$, and suppose that for each
$z\in\C\setminus\R$ the operator $B(z)$ belongs to $\K(\H_0,\H)$ and that the
difference $J[H_0,A_0]J^*-[H,A]$ belongs to $\K\big(\dom(H),\dom(H)^*\big)$. Then,
for each $\eta\in C^\infty_{\rm c}(\R)$ the difference of bounded operators
$$
J[A_0,\eta(H_0)]J^*-[A,\eta(H)]
$$
belongs to $\K(\H)$.
\end{Proposition}

\begin{proof}
By taking the various hypotheses into account one gets for any $z\in\C\setminus\R$
that
\begin{align*}
&J[A_0,R_0(z)]J^*-[A,R(z)]\\
&=JR_0(z)[H_0,A_0]R_0(z)J^*-R(z)[H,A]R(z)\\
&=\big\{JR_0(z)-R(z)J\big\}[H_0,A_0]R_0(z)J^*
+R(z)J[H_0,A_0]\big\{R_0(z)J^*-J^*R(z)\big\}\\
&\qquad+R(z)\big\{J[H_0,A_0]J^*-[H,A]\big\}R(z)\\
&=B(z)[H_0,A_0]R_0(z)J^*+R(z)J[H_0,A_0]B(\bar z)^*
+R(z)\big\{J[H_0,A_0]J^*-[H,A]\big\}R(z),
\end{align*}
with each term on the last line in $\K(\H)$. Now, by taking the first resolvent
identity into account, one obtains
$$
B(z)[H_0,A_0]R_0(z)J^*
=\big\{1+(z-i)R(z)\big\}B(i)\big\{1+(z-i)R_0(z)\big\}[H_0,A_0]R_0(i)
\big\{1+(z-i)R_0(z)\big\}J^*
$$
and
$$
R(z)J[H_0,A_0]B(\bar z)^*
=\big\{1+(z-i)R(z)\big\}R(i)J[H_0,A_0]\big\{1+(z-i)R_0(z)\big\}B(-i)^*
\big\{1+(z-i)R(z) \big\}
$$
as well as
\begin{align*}
&R(z)\big\{J[H_0,A_0]J^*-[H,A]\big\}R(z)\\
&=\big\{1+(z-i)R(z)\big\}R(i)\big\{J[H_0,A_0]J^*-[H,A]\big\}R(i)
\big\{1+(z-i)R(z)\big\}.
\end{align*}
Thus, by letting $z=\lambda+i\mu$ with $|\mu|\le 1$, one gets
the bound
$$
\big\|J[A_0,R_0(z)]J^*-[A,R(z)]\big\|_{\B(\H)}
\le{\rm Const.}\;\!\bigg(1+\frac{|\lambda+i(\mu-1)|}{|\mu|}\bigg)^3.
$$
One concludes as in the proof of Proposition \ref{ass2_short} by expressing
$J[A_0,\eta(H_0)]J^*-[A,\eta(H)]$ in terms of $J[A_0,R_0(z)]J^*-[A,R(z)]$ (using for
example \cite[Eq.~6.2.16]{ABG}), and then by dealing with a finite number of norm
convergent integrals of compact operators.
\end{proof}

As mentioned before the statement, no direct relation between $A_0$ and $A$ has been
imposed, and thus considering the special case $A=JA_0J^*$ is not really relevant.
However, it is not difficult to check how the quantity $J[H_0,A_0]J^*-[H,A]$ looks
like in that special case, and in applications such an approach could be of interest.
However, since the resulting formulas are rather involved in general, we do not
further investigate in that direction.

\section{One illustrative example}\label{Example}
\setcounter{equation}{0}

To illustrate our approach, we present below a simple example for which all the
computations can be made by hand (more involved examples will be presented elsewhere,
like in \cite{IRT11}, where part of the results of the present paper is used). In
this model, usually called one-dimensional Schr\"odinger operator with steplike
potential, the choice of a conjugate operator is rather natural, whereas the
computation of the $\varrho$-functions is not completely trivial due to the
anisotropy of the potential. We refer to \cite{Akt99,AJ07,Chr06,CK85,Ges86} for
earlier works on that model and to \cite{Ric05} for a $n$-dimensional generalisation.

So, we consider in the Hilbert space $\H:=\ltwo(\R)$ the Schr\"odinger operator
$H:=-\Delta+V$, where $V$ is the operator of multiplication by a function
$v\in C(\R;\R)$ with finite limits $v_\pm$ at infinity, \ie
$v_\pm:=\lim_{x\to\pm\infty}v(x)\in\R$. The operator $H$ is self-adjoint on
$\H^2(\R)$, since $V$ is bounded. As a second operator, we consider in the auxiliary
Hilbert space $\H_0:=\ltwo(\R)\oplus\ltwo(\R)$ the operator
$$
H_0:=(-\Delta+\vg)\oplus(-\Delta+\vd),
$$
which is also self-adjoint on its natural domain $\H^2(\R)\oplus\H^2(\R)$. Then, we
take a function $j_+\in C^\infty(\R;[0,1])$ with $j_+(x)=0$ if $x\le1$ and $j_+(x)=1$
if $x\ge2$, we set $j_-(x):=j_+(-x)$ for each $x\in\R$, and we define the
identification operator $J\in\B(\H_0,\H)$ by the formula
$$
J(\varphi_-,\varphi_+):=j_-\varphi_-+j_+\varphi_+,\quad(\varphi_-,\varphi_+)\in\H_0.
$$
Clearly, the adjoint operator $J^*\in\B(\H,\H_0)$ is given by
$J^*\psi=(j_-\psi,j_+\psi)$ for any $\psi\in\H$, and the operator $JJ^*\in\B(\H)$ is
equal to the operator of multiplication by $j_-^2+j_+^2$.

Let us now come to the choice of the conjugate operators. For $H_0$, the most natural
choice consists in two copies of the generator of dilations on $\R$, that is,
$A_0:=(D,D)$ with $D$ the generator of the group
$$
\big(\e^{itD}\psi\big)(x):=\e^{t/2}\psi(\e^t x),\quad\psi\in\S(\R),~t,x\in\R,
$$
where $\S(\R)$ denotes the Schwartz space on $\R$. In such a case, the map \eqref{C1}
with $(H,A)$ replaced by $(H_0,A_0)$ is strongly of class $C^\infty$ in $\H_0$.
Moreover, the $\varrho$-functions can be computed explicitly (see
\cite[Sec.~8.3.5]{ABG} for a similar calculation in an abstract setting):

$$
\widetilde\varrho^{A_0}_{H_0}(\lambda)
=\varrho^{A_0}_{H_0}(\lambda)
=\begin{cases}
+\infty & \hbox{if }\,\lambda<\min\{\vg,\vd\}\\
2\big(\lambda-\min\{\vg,\vd\}\big)
& \hbox{if }\,\min\{\vg,\vd\}\le\lambda<\max\{\vg,\vd\}\\
2\big(\lambda-\max\{\vg,\vd\}\big) & \hbox{if }\,\lambda\ge\max\{\vg,\vd\}.
\end{cases}
$$

For the conjugate operator for $H$, two natural choices exist: either one can use
again the generator $D$ of dilations in $\H$, or one can use the (formal) operator
$JA_0J^*$ which appears naturally in our framework. Since the latter choice
illustrates better the general case, we opt here for this choice and just note that
the former choice would also be suitable and would lead to similar results. So, we
set $\D:=\S(\R)$ and $j:=j_-+j_+$, and then observe that $JA_0J^*$ is well-defined
and equal to
\begin{equation}\label{explicit}
JA_0J^*=jDj
\end{equation}
on $\D$. This equality, the fact that $j$ is of class $C^1(D)$, and
\cite[Lemma~7.2.15]{ABG}, imply that $JA_0J^*$ is essentially self-adjoint on $\D$. We denote by
$A$ the corresponding self-adjoint extension.

We are now in a position for applying results of the previous sections such as
Theorem \ref{fonctionrho}. First, recall that $H_0$ is of class $C^1(A_0)$ and
observe that the assumption (iv) of Theorem \ref{fonctionrho} is satisfied with the
operator $J$ introduced above. Similarly, one easily shows that the assumption (iii)
of Theorem \ref{fonctionrho} also holds. Indeed, as mentioned after the statement of
Theorem \ref{fonctionrho}, the assumption (iii) holds if one shows that
$B(z)\in\K(\H_0,\H)$ for each $z\in\C\setminus\R$. But, for any
$(\varphi_-,\varphi_+)\in\H_0$, a direct calculation shows that
$B(z)(\varphi_-,\varphi_+)=B_-(z)\varphi_-+B_+(z)\varphi_+$, with
$$
B_\pm(z):=(H-z)^{-1}\big\{[-\Delta,j_\pm]+j_\pm(V-v_\pm)\big\}(-\Delta+v_\pm-z)^{-1}
\in\K(\H).
$$
So, one readily concludes that $B(z)\in\K(\H_0,\H)$.

Thus, one is only left with showing the assumption (ii) of Theorem \ref{fonctionrho}
and the $C^1(A)$-regularity of $H$. We first consider a short-range type
perturbation. In such a case, with $A$ defined as above, we know it is enough to
check the condition \eqref{cond_cor_comp} of Corollary \ref{compact}. For that
purpose, we assume the following stronger condition on $v:$
\begin{equation}\label{set1}
\lim_{|x|\to\infty}|x|\big(v(x)-v_\pm\big)=0,
\end{equation}
and observe that for each $(\varphi_-,\varphi_+)\in\S(\R)\oplus\S(\R)$ and
$z\in\C\setminus\R$ we have the equality
$$
B(z)A_0(\varphi_-,\varphi_+)=B_-(z)D\varphi_-+B_+(z)D\varphi_+.
$$
Then, taking into account the expressions for $B_-(z)$ and $B_+(z)$ as well as the
above assumption on $v$, one proves easily that
$\overline{B_\pm(z)D\upharpoonright\dom(D)}\in\K(\H)$, which implies
\eqref{cond_cor_comp}. Collecting our results, we end up with:

\begin{Lemma}[Short-range case]
Assume that $v\in C(\R;\R)$ satisfies \eqref{set1}, then the operator $H$ is of class
$C^1(A)$ and one has $\widetilde\varrho_H^A\ge\widetilde\varrho_{H_0}^{A_0}$. In
particular, $A$ is conjugate to $H$ on $\,\R\setminus\{\vg,\vd\}$.
\end{Lemma}

We now consider a long-range type perturbation and thus show that the assumptions of
Proposition \ref{ass2_long} hold with $A$ defined as above. For that purpose, we
assume that $v\in C^1(\R;\R)$ and that
\begin{equation}\label{set2}
\lim_{|x|\to\infty}|x|\;\!v'(x)=0.
\end{equation}
Then, a standard computation taking the inclusion $(H-z)^{-1}\D\subset\dom(A)$ into
account shows that $H$ is of class $C^1(A)$ with
\begin{equation}\label{form_com}
[A,H]=\big[j(-i\nabla)\;\!\id\;\!j,-\Delta\big]-ij^2\;\!\id\;\!v'
+\frac i2\big[j^2,-\Delta\big],
\end{equation}
where $\id$ is the function $\R\ni x\mapsto x\in\R$. Then, using \eqref{set2} and
\eqref{form_com}, one infers that $J[H_0,A_0]J^*-[H,A]$ belongs to
$\K\big(\dom(H),\dom(H)^*\big)$. Furthermore, simple considerations show that $J$
extends to an element of $\B\big(\dom(H_0)^*,\dom(H)^*\big)$. These results, together
with the ones already obtained, permit to apply Proposition \ref{ass2_long}, and thus
to get:

\begin{Lemma}[Long-range case]
Assume that $v\in C^1(\R;\R)$ satisfies \eqref{set2}, then the operator $H$ is of
class $C^1(A)$ and one has $\widetilde\varrho_H^A\ge\widetilde\varrho_{H_0}^{A_0}$.
In particular, $A$ is conjugate to $H$ on $\,\R\setminus\{\vg,\vd\}$.
\end{Lemma}

\section{Completeness of the wave operators}\label{Sec_Comp}
\setcounter{equation}{0}

One of the main goal in scattering theory is the proof of the completeness of the
wave operators. In our setting, this amounts to show that the strong limits
\begin{equation}\label{principal}
W_\pm(H,H_0,J):=\slim_{t\to\pm\infty}\e^{itH}J\e^{-itH_0}P_{\rm ac}(H_0)
\end{equation}
exist and have ranges equal to $\H_{\rm ac}(H)$. If the wave operators
$W_\pm(H,H_0,J)$ are partial isometries with initial sets $\H_0^\pm$, this implies in
particular that the scattering operator
$$
S:=W_+(H,H_0,J)^*\;\!W_-(H,H_0,J)
$$
is well-defined and unitary from $\H_0^-$ to $\H_0^+$.

When defining the completeness of the wave operators, one usually requires that
$\H_0^\pm=\H_{\rm ac}(H_0)$ (see for example \cite[Def.~III.9.24]{BW83} or
\cite[Def.~2.3.1]{Yaf92}). However, in applications it may happen that the ranges of
$W_\pm(H,H_0,J)$ are equal to $\H_{\rm ac}(H)$ but that
$\H_0^\pm\neq\H_{\rm ac}(H_0)$. Typically, this happens for multichannel type
scattering processes. In such situations, the usual criteria for completeness, as 
\cite[Prop.~III.9.40]{BW83} or \cite[Thm.~2.3.6]{Yaf92}, cannot be applied. So, we
present below a result about the completeness of the wave operators without assuming
that $\H_0^\pm=\H_{\rm ac}(H_0)$. Its proof is inspired by \cite[Thm.~2.3.6]{Yaf92}.

\begin{Proposition}\label{comp_serge}
Suppose that the wave operators defined in \eqref{principal} exist and are partial
isometries with initial set projections $P_0^\pm$. If there exists
$\widetilde J\in\B(\H,\H_0)$ such that
\begin{equation}\label{tralalere}
W_\pm\big(H_0,H,\widetilde J\big)
:=\slim_{t\to\pm\infty}\e^{itH_0}\widetilde J\e^{-itH}P_{\rm ac}(H)
\end{equation}
exist and such that
\begin{equation}\label{froufrou}
\slim_{t\to\pm\infty}\big(J\widetilde J-1\big)\e^{-itH}P_{\rm ac}(H)=0,
\end{equation}
then the equalities $\Ran\big(W_\pm(H,H_0,J)\big)=\H_{\rm ac}(H)$ hold. Conversely,
if $\Ran\big(W_\pm(H,H_0,J)\big)=\H_{\rm ac}(H)$ and if there exists
$\widetilde J\in\B(\H,\H_0)$ such that
\begin{equation}\label{trainversTokyo}
\slim_{t\to\pm\infty}\big(\widetilde JJ-1\big)\e^{-itH_0}P_0^\pm=0,
\end{equation}
then $W_\pm\big(H_0,H,\widetilde J\big)$ exist and \eqref{froufrou} holds.
\end{Proposition}

\begin{proof}
(i) By using the chain rule for wave operators \cite[Thm.~2.1.7]{Yaf92}, we deduce
from the definitions \eqref{principal}-\eqref{tralalere} that the limits
$$
W_\pm \big(H,H,J\widetilde J\big)
:=\slim_{t\to\pm\infty}\e^{itH}J\widetilde J\e^{-itH}P_{\rm ac}(H)
$$
exist and satisfy
\begin{equation}\label{prod_wave}
W_\pm\big(H,H,J\widetilde J\big)= W_\pm(H,H_0,J)W_\pm\big(H_0,H,\widetilde J\big).
\end{equation}
In consequence, the equality
$$
\slim_{t\to\pm\infty}\big(\e^{itH}J\widetilde J\e^{-itH}P_{\rm ac}(H)
-P_{\rm ac}(H)\big)=0,
$$
which follow from \eqref{froufrou}, implies that
$W_\pm\big(H,H,J\widetilde J\big)P_{\rm ac}(H)=P_{\rm ac}(H)$. This, together with
\eqref{prod_wave} and the equality
$W_\pm\big(H_0,H,\widetilde J\big)=W_\pm\big(H_0,H,\widetilde J\big)P_{\rm ac}(H)$,
gives
$$
W_\pm(H,H_0,J)W_\pm\big(H_0,H,\widetilde J\big)
=W_\pm\big(H,H,J\widetilde J\big)P_{\rm ac}(H)
=P_{\rm ac}(H),
$$
which is equivalent to
$$
W_\pm\big(H_0,H,\widetilde J\big)^*W_\pm(H,H_0,J)^*=P_{\rm ac}(H).
$$
This gives the inclusion
$
\Ker\big(W_\pm(H,H_0,J)^*\big)\subset\H_{\rm ac}(H)^\perp
$,
which together with the fact that the range of a partial isometry is closed imply
that
$$
\H=\Ran\big(W_\pm(H,H_0,J)\big)\oplus\Ker\big(W_\pm(H,H_0,J)^*\big)
\subset\H_{\rm ac}(H)\oplus\H_{\rm ac}(H)^\perp
=\H.
$$
So, one must have $\Ran\big(W_\pm(H,H_0,J)\big)=\H_{\rm ac}(H)$, and the first claim
is proved.

(ii) Conversely, consider $\psi\in\H_{\rm ac}(H)$. Then we know from the hypothesis
$\Ran\big(W_\pm(H,H_0,J)\big)=\H_{\rm ac}(H)$ that there exist
$\psi_\pm\in P_0^\pm\H_0$ such that
\begin{equation}\label{reserved}
\lim_{t\to\pm\infty}\big\|\e^{-itH}\psi-J\e^{-itH_0}P_0^\pm\psi_\pm\big\|_\H=0.
\end{equation}
Together with \eqref{trainversTokyo}, this implies that the norm
\begin{align*}
&\big\|\e^{itH_0}\widetilde J\e^{-itH}\psi-P_0^\pm\psi_\pm \big\|_{\H_0}\\
&\le\big\|\e^{itH_0}\widetilde J
\big(\e^{-itH}\psi-J\e^{-itH_0}P_0^\pm\psi_\pm\big)\big\|_{\H_0}
+\big\|\e^{itH_0}\widetilde JJ\e^{-itH_0}P_0^\pm\psi_\pm
-P_0^\pm\psi_\pm\big\|_{\H_0}\\
&\le{\rm Const.}\;\!\big\|\e^{-itH}\psi-J\e^{-itH_0}P_0^\pm\psi_\pm\big\|_\H
+\big\|\big(\widetilde J J-1\big)\e^{-itH_0}P_0^\pm \psi_\pm \big\|_{\H_0}
\end{align*}
converges to $0$ as $t\to\pm\infty$, showing that the wave operators
\eqref{tralalere} exist.

For the relation \eqref{froufrou}, observe first that \eqref{trainversTokyo} gives
$$
\slim_{t\to\pm\infty}\big(J\widetilde J-1\big)J\e^{-itH_0}P_0^\pm
=\slim_{t\to\pm\infty}J\big(\widetilde JJ-1\big)\e^{-itH_0}P_0^\pm
=0.
$$
Together with \eqref{reserved}, this implies that the norm
\begin{align*}
&\big\|\big(J\widetilde J-1\big)\e^{-itH}\psi\big\|_\H\\
&\le\big\|\big(J\widetilde J-1\big)\big(J\e^{-itH_0}P_0^\pm\psi_\pm
-\e^{-itH}\psi\big)\big\|_\H
+\big\|\big(J\widetilde J-1\big)J\e^{-itH_0}P_0^\pm\psi_\pm\big\|_\H\\
&\le{\rm Const.}\;\!\big\|\e^{-itH}\psi-J\e^{-itH_0}P_0^\pm\psi_\pm\big\|_\H
+\big\|\big(J\widetilde J-1\big)J\e^{-itH_0}P_0^\pm\psi_\pm\big\|_\H
\end{align*}
converges to $0$ as $t \to \pm \infty$, showing that \eqref{froufrou} also holds.
\end{proof}


\def\cprime{$'$}

\end{document}